\documentclass[11pt]{article}

\usepackage{a4}
\usepackage{authblk}
\usepackage{amsmath,amssymb}
\usepackage{amsthm}
\usepackage[dvipdfmx]{graphicx}

\pdfoutput=1

\usepackage{thmtools}

\newtheorem{theorem}{Theorem}
\newtheorem{lemma}{Lemma}
\newtheorem{corollary}{Corollary}

\theoremstyle{definition}

\newtheorem{problem}{Problem}
\newtheorem{remark}{Remark}

\newtheorem{observation}{Observation}

\usepackage{array}
\usepackage{ascmac}
\usepackage{booktabs}
\usepackage{wrapfig}
\usepackage{url}
\usepackage{here}

\usepackage[final,mode=multiuser]{fixme}
\fxuseenvlayout{colorsig}
\fxusetargetlayout{color}
\FXRegisterAuthor{si}{asi}{SI}
\fxsetface{env}{}
\newcommand{\Prefix}{\mathsf{Prefix}}
\newcommand{\Substr}{\mathsf{Substr}}
\newcommand{\Suffix}{\mathsf{Suffix}}

\newcommand{\lspo}{\mathsf{lspo}}

\newcommand{\trie}{\mathsf{Trie}}
\newcommand{\str}{\mathsf{str}}
\newcommand{\flink}{\mathsf{flink}}

\newcommand{\ctrie}{\mathsf{ComTrie}}
\newcommand{\des}{\mathsf{des}}

\newcommand{\FLT}{\mathsf{FLTree}}
\newcommand{\AC}{\mathsf{AC}}

\newcommand{\rev}[1]{#1^R}

\newcommand{\DAWG}{\mathsf{DAWG}}
\newcommand{\EndPos}{\mathsf{End\_Pos}}
\newcommand{\Eqc}[1]{[{#1}]}
\newcommand{\Long}{\mathsf{long}}
\newcommand{\LPT}{\mathsf{LPTree}}
\newcommand{\slink}{\mathsf{slink}}
\newcommand{\SLT}{\mathsf{SLTree}}

\newcommand{\STree}{\mathsf{STree}}

\newcommand{\ind}{\mathsf{ind}}

\begin{document}

\title{All-Pairs Suffix-Prefix on Fully Dynamic Set of Strings}

\author[1]{Masaru Kikuchi}
\author[2]{Shunsuke~Inenaga}

\affil[1]{Department of Information Science and Technology
{\tt kikuchi.masaru.484@s.kyushu-u.ac.jp}}

\affil[2]{Department of Informatics, Kyushu University, Japan
{\tt inenaga.shunsuke.380@m.kyushu-u.ac.jp}}

\date{}
\maketitle

\begin{abstract}
The \emph{all-pairs suffix-prefix} (\emph{APSP}) problem
is a classical problem in string processing which has important applications in bioinformatics.
Given a set $\mathcal{S} = \{S_1, \ldots, S_k\}$ of $k$ strings,
the APSP problem asks one to compute the longest suffix of $S_i$ that is a prefix of $S_j$
for all $k^2$ ordered pairs $\langle S_i, S_j \rangle$ of strings in $\mathcal{S}$.
In this paper, we consider the \emph{dynamic} version of the APSP problem
that allows for insertions of new strings to the set of strings.
Our objective is, each time a new string $S_i$ arrives to
the current set $\mathcal{S}_{i-1} = \{S_1, \ldots, S_{i-1}\}$ of $i-1$ strings,
to compute (1) the longest suffix of $S_i$ that is a prefix of $S_j$
and (2) the longest prefix of $S_i$ that is a suffix of $S_j$ for all $1 \leq j \leq i$.
We propose an $O(n)$-space data structure
which computes (1) and (2) in $O(|S_i| \log \sigma + i)$ time for each new given string $S_i$,
where $n$ is the total length of the strings
and $\sigma$ is the alphabet size.
Further, we show how to extend our methods to the \emph{fully dynamic}
version of the APSP problem allowing for both insertions and deletions of strings.
\end{abstract}

\section{Introduction}

The \emph{all-pairs suffix-prefix} (\emph{APSP}) problem
is a classical problem in string processing which has important applications in bioinformatics,
since it is the first step of genome assembly~\cite{Gusfield1997}.
Given a set $\mathcal{S} = \{S_1, \ldots, S_k\}$ of $k$ strings,
the APSP problem asks one to compute the longest suffix of $S_i$ that is a prefix of $S_j$
for all $k^2$ ordered pairs $\langle S_i, S_j \rangle$ of strings in $\mathcal{S}$.

A straightforward solution, mentioned in~\cite{GusfieldLS92},
is to use a variant of the KMP pattern matching algorithm~\cite{KnuthMP77}
for each pair of two strings $S_i, S_j$ in $O(|S_i|+|S_j|)$ time.
This however leads to an inefficient $O(kn)$-time complexity,
where $n = \Vert \mathcal{S} \Vert$ is the total length of the strings in $\mathcal{S}$.

Gusfield et al.~\cite{GusfieldLS92} proposed the first efficient solution
to the APSP problem that takes $O(n)$ space,
which is based on the \emph{generalized suffix tree}~\cite{Blumer1987,GusfieldLS92} for the set $\mathcal{S}$ of strings.
After building the generalized suffix tree
in $O(n \log \sigma)$ time in the case of a general ordered alphabet of size $\sigma$~\cite{Weiner73}
or in $O(n)$ time in the case of an integer alphabets of size $\sigma = n^{O(1)}$~\cite{Farach-ColtonFM00},
Gusfield et al.'s algorithm~\cite{GusfieldLS92} works in $O(n + k^2)$ optimal time.
Ohlebusch and Gog~\cite{OhlebuschG10}
proposed an alternative algorithm
for solving the APSP problem with the same complexities as above,
using (enhanced) suffix arrays~\cite{ManberM93,AbouelhodaKO04}.
Tustumi et a.~\cite{TustumiGTL16} gave an improved suffix-array based algorithm
that is fast and memory efficient in practice.
All these approaches share the common concepts of building
the generalized suffix tree/array for $\mathcal{S}$.

Another data structure that can be used to solve the APSP problem
is the \emph{Aho-Corasick automata} (\emph{AC-automata})~\cite{Aho1975StringMatching}.
This is intuitive since
the AC-automaton is a generalization of the KMP-automaton for multiple strings.
The use of the AC-automaton for solving APSP was suggested by Lim and Park~\cite{LimP17},
but they did not give any details of an algorithm nor the complexity.
Recently, Loukides and Pissis~\cite{LoukidesP22} proposed an AC-automaton based algorithm
for the APSP problem.
After building the AC-automaton 
in $O(n \log \sigma)$ time in the case of a general ordered alphabet of size $\sigma$~\cite{Aho1975StringMatching}
or in $O(n)$ time in the case of an integer alphabets of size $\sigma = n^{O(1)}$~\cite{DoriL06},
the algorithm of Loukides and Pissis~\cite{LoukidesP22} runs in optimal $O(n+k^2)$ time.
Their algorithm can also solve a length-threshold version of the problem in optimal time.
Recently, Loukides et al.~\cite{LoukidesPTZ23} considered
a query-version of the APSP problem.
They presented a data structure of $O(n)$-space that can
report the longest suffix-prefix match between one string $S_i$
and all the other strings in $O(k)$ time.
Their data structure is based on the AC-automaton
and the micro-macro decomposition of trees~\cite{AlstrupHLT97}.

In this paper, we consider the \emph{dynamic} version of the APSP problem
that allows for insertions of new strings to the set of strings.
Our objective is, each time a new string $S_i$ arrives to
the current set $\mathcal{S}_i = \{S_1, \ldots, S_{i-1}\}$ of strings,
to compute the following longest suffix-prefix matches between a new string $S_i$
and all the other strings $S_1, \ldots, S_{i-1}$:
\begin{itemize}
\item the longest suffix of $S_i$ that is a prefix of $S_j$ for $1 \leq j \leq i$, and
\item the longest prefix of $S_i$ that is a suffix of $S_j$ for $1 \leq j \leq i$.
\end{itemize}
We propose an $O(n)$-space data structure based on the
\emph{directed acyclic word graph} (\emph{DAWG}) for a growing set of strings,
which is able to compute the suffix-prefix matches for a given new string $S_i$
in $O(|S_i| \log \sigma + i)$ time.
After iterating this for all $k$ strings $S_1, \ldots, S_k$ arriving to $\mathcal{S}$,
the total running time becomes $O(n \log \sigma + k^2)$,
which is as fast as the state-of-the-art algorithms~\cite{GusfieldLS92,LoukidesP22}
for a static set of strings
in the case of general ordered alphabets of size $\sigma$.
We note that the suffix-tree based algorithm of Gusfield et al.~\cite{GusfieldLS92} performs a DFS, and thus, does not seem to extend directly to the dynamic case.

Further, we consider the \emph{fully dynamic} APSP problem
which allows for both insertions and deletions of strings.
We show how to solve fully dynamic APSP 
in $O(|S_i|\log \sigma + k)$ amortized time when string $S_i$ is inserted or deleted, where $k$ is the number of strings in the current set $\mathcal{S}$ of strings.
In this fully dynamic version, we show how 
the DAWG-based algorithm can be simulated with the suffix tree in the case of insertions,
and how our suffix-tree based method can handle the case of deletions.
We emphasize that our approach is again quite different from
Gusfield et al.'s method~\cite{GusfieldLS92}.

Before describing our APSP algorithms for the (fully) dynamic case (Section~\ref{sec:dynamic} and Section~\ref{sec:fully_dynamic}),
we present a new AC-automata based APSP algorithm for a static set of strings (Section~\ref{sec:static}).
While the algorithm of Loukides and Pissis~\cite{LoukidesP22}
is based on a reduction of the APSP problem to the ULIT (union of labeled intervals on a tree) problem,
our new APSP algorithm does not use this reduction.
Instead, we use simple failure link traversals of the AC-automaton,
and marking operations on
a compact prefix trie that is a compact version of the AC-trie.
The simple nature of our algorithm makes it possible to
work also on the DAWG and on the suffix tree in the (fully) dynamic case.

The \emph{hierarchical overlap graphs} (\emph{HOGs}) and their variants
have a close relationship to the APSP problem.
There are AC-automata oriented approaches for
computing these graphs~\cite{CanovasCR17,Khan21,ParkPCPR21}.
We remark that all of their approaches are for the static case,
and do not seem to immediately extend to the dynamic setting.

A preliminary version of this paper appeared in~\cite{KikuchiI24}.
The new material in this full version is
the algorithm for the fully dynamic APSP problem (Section~\ref{sec:fully_dynamic}),
which was posed as an open question in the preliminary version~\cite{KikuchiI24}.

\section{Preliminaries} \label{sec:preliminaries}

\subsection{Strings}
Let $\Sigma$ denote an ordered \emph{alphabet} of size $\sigma$.
An element of $\Sigma^*$ is called a \emph{string}.
The length of a string $S \in \Sigma^*$ is denoted by $|S|$.
The \emph{empty string} $\varepsilon$ is the string of length $0$.
Let $\Sigma^+ = \Sigma^* \setminus \{\varepsilon\}$.
For string $S = xyz$, $x$, $y$, and $z$ are called
the \emph{prefix}, \emph{substring}, and \emph{suffix} of $S$,
respectively.
Let $\Prefix(S)$, $\Substr(S)$, and $\Suffix(S)$ denote
the sets of prefixes, substrings, and suffixes of $S$, respectively.
The elements of $\Prefix(S) \setminus \{S\}$,
$\Substr(S) \setminus \{S\}$, and $\Suffix(S) \setminus \{S\}$ are called
the \emph{proper prefixes}, \emph{proper substrings},
and \emph{proper suffixes} of $S$, respectively.
For a string $S$ of length $n$, $S[i]$ denotes the $i$-th symbol of $S$
and $S[i..j] = S[i] \cdots S[j]$ denotes the substring of $S$
that begins at position $i$ and ends at position $j$ for $1 \leq i \leq j \leq n$.
For convenience, let $S[i..j] = \varepsilon$ for $i > j$.
The \emph{reversed string} of a string $S$ is denoted by $\rev{S}$,
that is, $\rev{S} = S[|S|] \cdots S[1]$.

For a set $\mathcal{S} = \{S_1, \ldots, S_k\}$ of strings,
let $\Substr(\mathcal{S}) = \bigcup_{i=1}^k \Substr(S_i)$.
Let $\Vert \mathcal{S} \Vert = \sum_{i=1}^k |S_i|$ denote
the total length of the strings in $\mathcal{S} = \{S_1, \ldots, S_k\}$.


\subsection{All-Pairs Suffix-Prefix Overlap (APSP) Problems}

In this section, we introduce the problems that we tackle in this paper.

Let $\mathcal{S} = \{S_1, \ldots, S_k\}$ be a set of non-empty strings,
where $S_i \in \Sigma^+$ for $1 \leq i \leq k$.
For each ordered pair $\langle S_i, S_j \rangle$ of two strings in $\mathcal{S}$,
the longest string in the set $\Suffix(S_i) \cap \Prefix(S_j)$
is called the \emph{longest suffix-prefix overlap} of $S_i$ and $S_j$,
and is denoted by $\lspo(i,j)$.
We encode each $\lspo(i,j)$ by a tuple $(i,j,|\lspo(i,j)|)$ using $O(1)$ space.

\subsubsection{APSP Problems on Static Sets}

\begin{problem}[static-APSP] \label{prob:static}
The \emph{static all-pairs suffix-prefix overlaps} (\emph{static-APSP}) problem is,
given a static set $\mathcal{S} = \{S_1, \ldots, S_k\}$ of $k$ non-empty strings, to compute
\begin{equation*}
  \mathcal{L_S} = \{\lspo(i,j) \mid 1 \leq i \leq k, 1 \leq j \leq k\}.
\end{equation*}
\end{problem}
Note that the output size of Problem~\ref{prob:static} is $|\mathcal{L_S}| = \Theta(k^2)$
as we compute $\lspo(i,j)$ for all ordered pairs $\langle S_i,S_j \rangle$ of strings in $\mathcal{S}$.

\begin{problem}[length-bounded static-APSP] \label{prob:static_length}
The \emph{length-bounded} static-APSP problem is,
given a static set $\mathcal{S} = \{S_1, \ldots, S_k\}$ of $k$ non-empty strings
and an integer threshold $\ell \geq 0$, compute
\[
\mathcal{L_S^{\geq \ell}} = \{\lspo(i,j) \mid 1 \leq i \leq k, 1 \leq j \leq k, |\lspo(i,j)| \geq \ell \}.
\]
\end{problem}

\subsubsection{APSP on Dynamic Sets with Insertions}

We consider a \emph{dynamic} variant of APSP allowing for insertions of strings, named \emph{dynamic APSP}, that is defined as follows:

\begin{problem}[dynamic APSP] \label{prob:dynamic}
The \emph{dynamic} APSP problem is, each time a new string $S_{i}$ is inserted to the
current set $\mathcal{S} = \{S_1, \ldots, S_{i-1}\}$ of $i-1$ strings,
to compute the two following sets 
\begin{eqnarray*}
  \mathcal{F}_i & = & \{\lspo(i,j) \mid 1 \leq j \leq i\},  \\
  \mathcal{B}_i  & = & \{\lspo(j,i) \mid 1 \leq j \leq i\}.
\end{eqnarray*}
\end{problem}
We remark that $\mathcal{L_S} = \bigcup_{i=1}^k \left( \mathcal{F}_i \cup \mathcal{B}_i \right)$ holds for the set $\mathcal{S} = \{S_1, \ldots, S_k\}$ of $k$ strings.
Thus, a solution to Problem~\ref{prob:dynamic} is also a solution to Problem~\ref{prob:static}.

\begin{problem}[length-bounded dynamic APSP] \label{prob:dynamic_length}
Let $\ell > 0$ be a length threshold.
The \emph{dynamic}-APSP problem is, each time a new string $S_{i}$ is inserted to the
current set $\mathcal{S} = \{S_1, \ldots, S_{i-1}\}$ of $i-1$ strings,
to compute the two following sets 
\begin{eqnarray*}
  \mathcal{F}_i^{\geq \ell} & = & \{\lspo(i,j) \mid 1 \leq j \leq i, |\lspo(i,j)| \geq \ell\},  \\
  \mathcal{B}_i^{\geq \ell}  & = & \{\lspo(j,i) \mid 1 \leq j \leq i, |\lspo(j,i)| \geq \ell \}.
\end{eqnarray*}
\end{problem}
Analogously, $\mathcal{L_S}^{\geq \ell} = \bigcup_{i=1}^k \left( \mathcal{F}_i^{\geq \ell} \cup \mathcal{B}_i^{\geq \ell} \right)$ holds for the set $\mathcal{S} = \{S_1, \ldots, S_k\}$ of $k$ strings.
Thus, a solution to Problem~\ref{prob:dynamic_length} is also a solution to Problem~\ref{prob:static_length} as well.

\subsubsection{APSP on Fully Dynamic Sets with Insertions and Deletions}

We also consider a \emph{fully dynamic} variant of APSP allowing for insertions and deletions of strings, named \emph{fully dynamic APSP}, that is defined as follows:

\begin{problem}[fully dynamic APSP] \label{prob:dynamic}
In the \emph{fully dynamic} APSP problem, when a new string $S_{i}$ is added to the
current set $\mathcal{S}$ of strings, the goal is to compute the two following sets 
\begin{eqnarray*}
  \mathcal{F}_i & = & \{\lspo(i,j) \mid j \in \ind(\mathcal{S})\},  \\
  \mathcal{B}_i  & = & \{\lspo(j,i) \mid j \in \ind(\mathcal{S})\},
\end{eqnarray*}
and add them to the list $\mathcal{L}_{\mathcal{S}}$,
where $\ind(\mathcal{S})$ denotes the set of indices $h$ such that $S_h \in \mathcal{S}$.
When a string $S_i$ is deleted from the current set $\mathcal{S}$ of strings,
then the goal is to remove $\mathcal{F}_i$ and $\mathcal{B}_i$ from the list $\mathcal{L}_{\mathcal{S}}$.
\end{problem}

\begin{problem}[length-bounded fully dynamic APSP] \label{prob:dynamic_length}
Let $\ell > 0$ be a length threshold.
In the \emph{length-bounded fully dynamic} APSP problem, when a new string $S_{i}$ is inserted to the
current set $\mathcal{S}$ of strings,
to compute the two following sets 
\begin{eqnarray*}
  \mathcal{F}_i^{\geq \ell} & = & \{\lspo(i,j) \mid 1 \leq j \leq i, |\lspo(i,j)| \geq \ell\},  \\
  \mathcal{B}_i^{\geq \ell}  & = & \{\lspo(j,i) \mid 1 \leq j \leq i, |\lspo(j,i)| \geq \ell \},
\end{eqnarray*}
and add them to the list $\mathcal{L}_{\mathcal{S}}^{\geq \ell}$,
where $\ind(\mathcal{S})$ denotes the set of indices $h$ such that $S_h \in \mathcal{S}$.
When a string $S_i$ is deleted from the current set $\mathcal{S}$ of strings,
then the goal is to remove $\mathcal{F}_i^{\geq \ell}$ and $\mathcal{B}_i^{\geq \ell}$ from the list $\mathcal{L}_{\mathcal{S}}^{\geq \ell}$.
\end{problem}

\section{Tools}

In this section we list the data structures that are used as
building blocks of our static-APSP and dynamic-APSP algorithms.

\subsection{Tries and Compact Tries}

A \emph{trie} $\mathsf{T}$ is a rooted tree $(\mathsf{V}, \mathsf{E})$ such that 
\begin{itemize}
  \item each edge in $\mathsf{E}$ is labeled by a single character from $\Sigma$ and
  \item the character labels of the out-going edges of each node begin with distinct characters.
\end{itemize}

For a pair of nodes $u,v$ in a trie where $u$ is an ancestor of $v$,
let $\langle u, v \rangle$ denote the path from $u$ to $v$.
The \emph{path label} of $\langle u, v \rangle$
is the concatenation of the edge labels from $u$ to $v$.
The \emph{string label} of a node $v$ in a trie, denoted $\str(v)$,
is the path label from the root to $v$.
The \emph{string depth} of a node $v$ in a trie is the length
of the string label of $v$, that is $|\str(v)|$.

For a set $\mathcal{S} = \{S_1,\ldots, S_k\}$ of $k$ strings,
let $\trie(\mathcal{S})$ denote the trie that represents the strings in $\mathcal{S}$.
We assume that each of the $k$ nodes in $\trie(\mathcal{S})$
representing $S_i$ stores a unique integer $i$ called the \emph{id}.

We build and maintain a compacted trie $\ctrie(\mathcal{S})$ that is obtained from
$\trie(\mathcal{S})$ by keeping
\begin{itemize}
  \item all the leaves and the (possibly non-branching) internal nodes
    that store id's $1, \ldots, k$ for the $k$ strings $S_1, \ldots, S_k$, and
  \item the root and all the branching nodes.
\end{itemize}
All other other nodes, each having a single child,
are removed and their corresponding edges are contracted into paths.
For each node $v$ in $\trie(\mathcal{S})$,
let $\des(v)$ denote the shallowest descendant of $v$
that exists in $\ctrie(\mathcal{S})$.
Note that the number of nodes in $\ctrie(\mathcal{S})$ is $O(k)$.

\begin{figure}[t]
  \centering
  \includegraphics[scale=0.32]{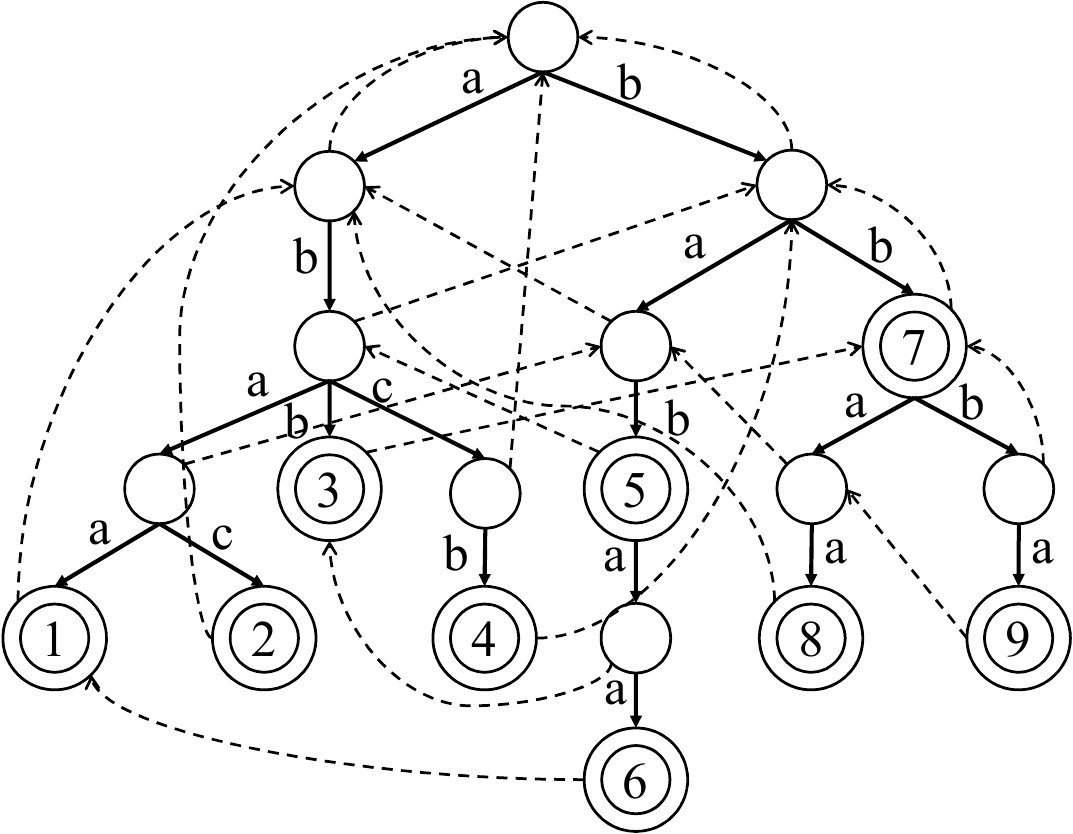}
  \hfill
  \includegraphics[scale=0.32]{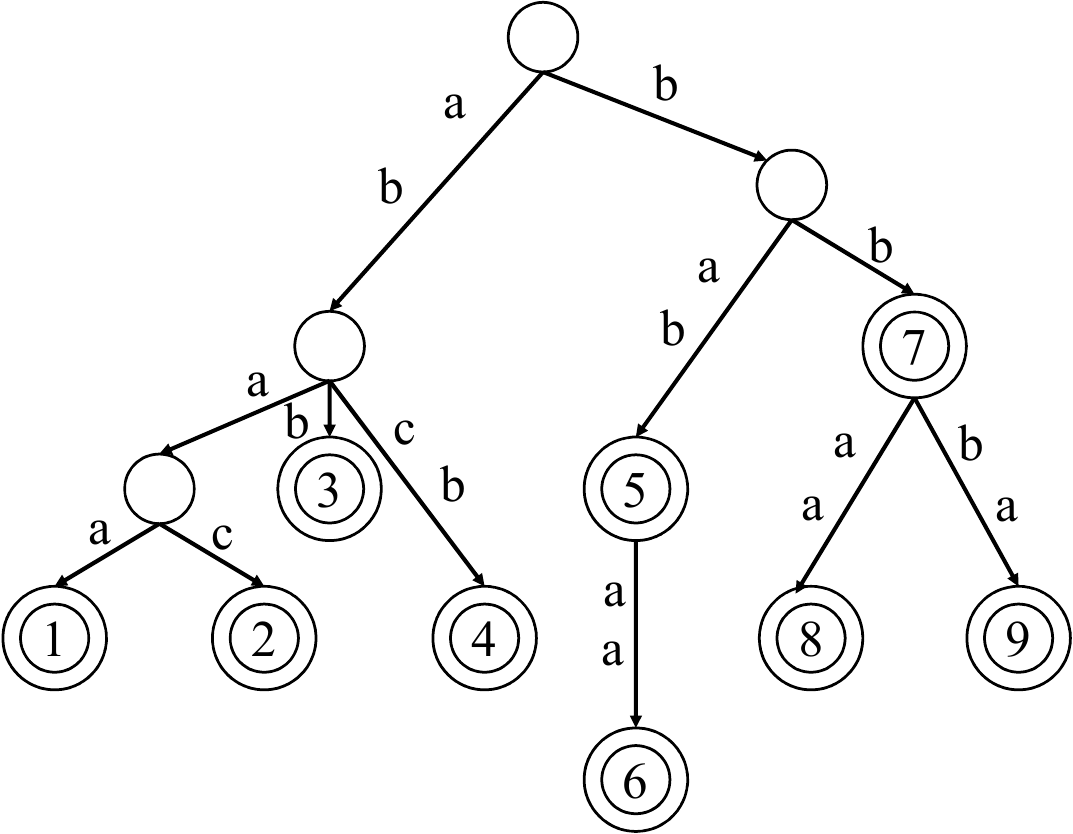}
  \caption{Illustrations of $\AC(\mathcal{S})$ (left) and $\ctrie(\mathcal{S})$ (right) for the set $\mathcal{S} = \{\mathrm{abaa, abac, abb, abcb, bab, babaa, bb, bbaa, bbba}\}$ of strings. 
    The bold solid arcs represent trie edges and the dashed arcs represent failure links.
  The nodes representing the strings in $\mathcal{S}$ are depicted by double-lined circles with the string id's.}
  \label{fig:AC}
\end{figure}

\subsection{Aho-Corasick Automata}

For each non-root node $v$ of $\trie(\mathcal{S})$,
define its \emph{failure link}
by $\flink(v) = u$ iff $\str(u)$ is the longest proper suffix of $\str(v)$
that can be spelled out from the root of $\trie(\mathcal{S})$.
In other words, $\str(u)$ is the longest prefix of some string(s) represented by $\trie(\mathcal{S})$ that is also a proper suffix of $\str(v)$.
The failure links also form a (reversed) tree, and let us denote it by $\FLT(\mathcal{S})$.

The \emph{Aho-Corasick automaton}~\cite{Aho1975StringMatching} for a set $\mathcal{S}$ of strings,
denoted $\AC(\mathcal{S})$, is a finite automaton with two kinds of transitions\footnote{We do not use the output function in our algorithms.}:
the goto function that is represented by the edges of the trie $\trie(\mathcal{S})$,
and the failure function that is represented by the edges of the reversed trie $\FLT(\mathcal{S})$.
It is clear that the total number of nodes and edges of $\AC(\mathcal{S})$
is $O(n)$, where $n = \Vert \mathcal{S} \Vert$.
See Fig.~\ref{fig:AC} for concrete examples of the AC-automaton
and its trie, compact trie, and failure links.

\begin{theorem}[\cite{DoriL06,Aho1975StringMatching}] \label{theo:AC_const}
  For a set $\mathcal{S}$ of strings of total length $n$,
  $\AC(\mathcal{S})$ can be built
  \begin{itemize}
  \item in $O(n)$ time and space for an integer alphabet of size $\sigma = n^{O(1)}$;
  \item in $O(n \log \sigma)$ time and $O(n)$ space
    for a general ordered alphabet of size $\sigma$.
  \end{itemize}
\end{theorem}


\subsection{Directed acyclic word graphs (DAWGs)}


For a set $\mathcal{S}$ of strings,
let $\EndPos_{\mathcal{S}}(w)$ denote the set of pairs of
string ids and ending positions of all occurrences of a substring $w \in \Substr(\mathcal{S})$,
that is,
\[\EndPos_{\mathcal{S}}(w) = \{(i, j) \mid S_i[j-|w|+1..j] = w, 1 \leq i \leq k, 1 \leq j \leq |S_i|\}.
\]
We consider an equivalence relation $\equiv_{\mathcal{S}}$ of strings over $\Sigma$ w.r.t. $\mathcal{S}$ such that, for any two strings $w$ and $u$,
$w \equiv_{\mathcal{S}} u$ iff $\EndPos_{\mathcal{S}}(w) = \EndPos_{\mathcal{S}}(u)$.
For any string $x \in \Sigma^*$,
let $[x]_{\mathcal{S}}$ denote the equivalence class for $x$ w.r.t. $\equiv_{\mathcal{S}}$.

The DAWG of a set $\mathcal{S}$ of strings, denoted $\DAWG(\mathcal{S})$,
  is an edge-labeled DAG $(V, E)$ such that
  \begin{eqnarray*}
    V & = & \{\Eqc{x}_{\mathcal{S}} \mid x \in \Substr(\mathcal{S})\}, \\
    E & = & \{(\Eqc{x}_{\mathcal{S}}, b, \Eqc{xb}_{\mathcal{S}}) \mid x, xb \in \Substr(\mathcal{S}), b \in \Sigma\}.
  \end{eqnarray*}
  We also define the set $L$ of {\em suffix links} of $\DAWG(\mathcal{S})$
  by
  \[
    L = \{(\Eqc{ax}_{\mathcal{S}}, a, \Eqc{x}_{\mathcal{S}}) \mid x, ax \in \Substr(\mathcal{S}), a \in \Sigma, \Eqc{ax}_{\mathcal{S}} \neq \Eqc{x}_{\mathcal{S}} \}.
  \]
Namely, two substrings $x$ and $y$ in $\Substr(\mathcal{S})$ are
represented by the same node of $\DAWG(\mathcal{S})$
iff the ending positions of $x$ and $y$ in the strings of $\mathcal{S}$ are equal.

See Fig.~\ref{fig:DAWG} and~\ref{fig:DAWG_suffixlinks} for
a concrete example of the DAWG and its suffix links.

\begin{figure}[h!]
  \centering
  \includegraphics[scale=0.32]{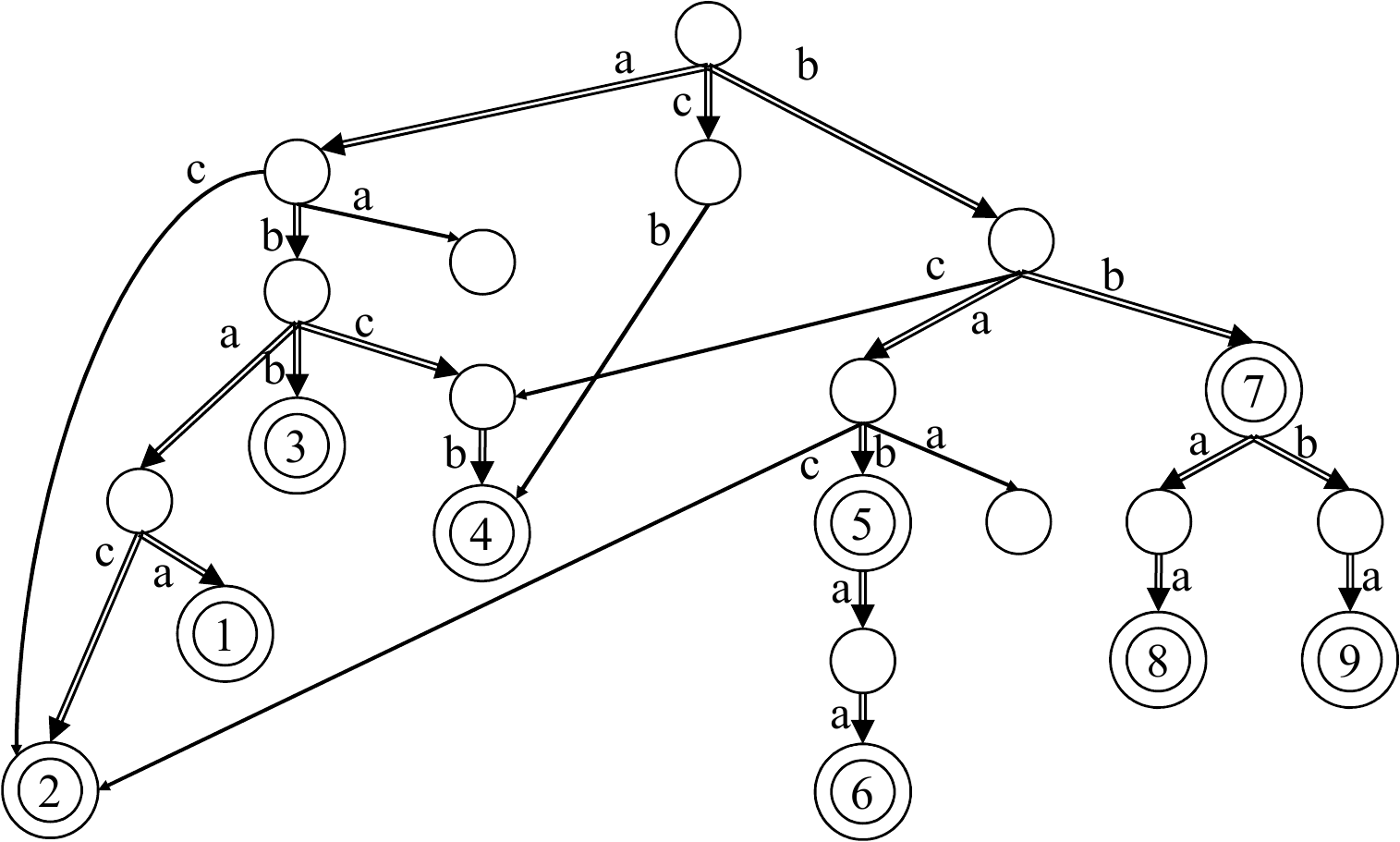}
  \caption{$\DAWG(\mathcal{S})$ for the same set $\mathcal{S} = \{\mathrm{abaa, abac, abb, abcb, bab, babaa, bb, bbaa, bbba}\}$ of strings as in Fig.~\ref{fig:AC}.
  The induced tree consisting only of the double-lined arcs is $\trie(\mathcal{S})$.}
  \label{fig:DAWG}
\end{figure}

\begin{figure}[h!]
  \centering
  \includegraphics[scale=0.32]{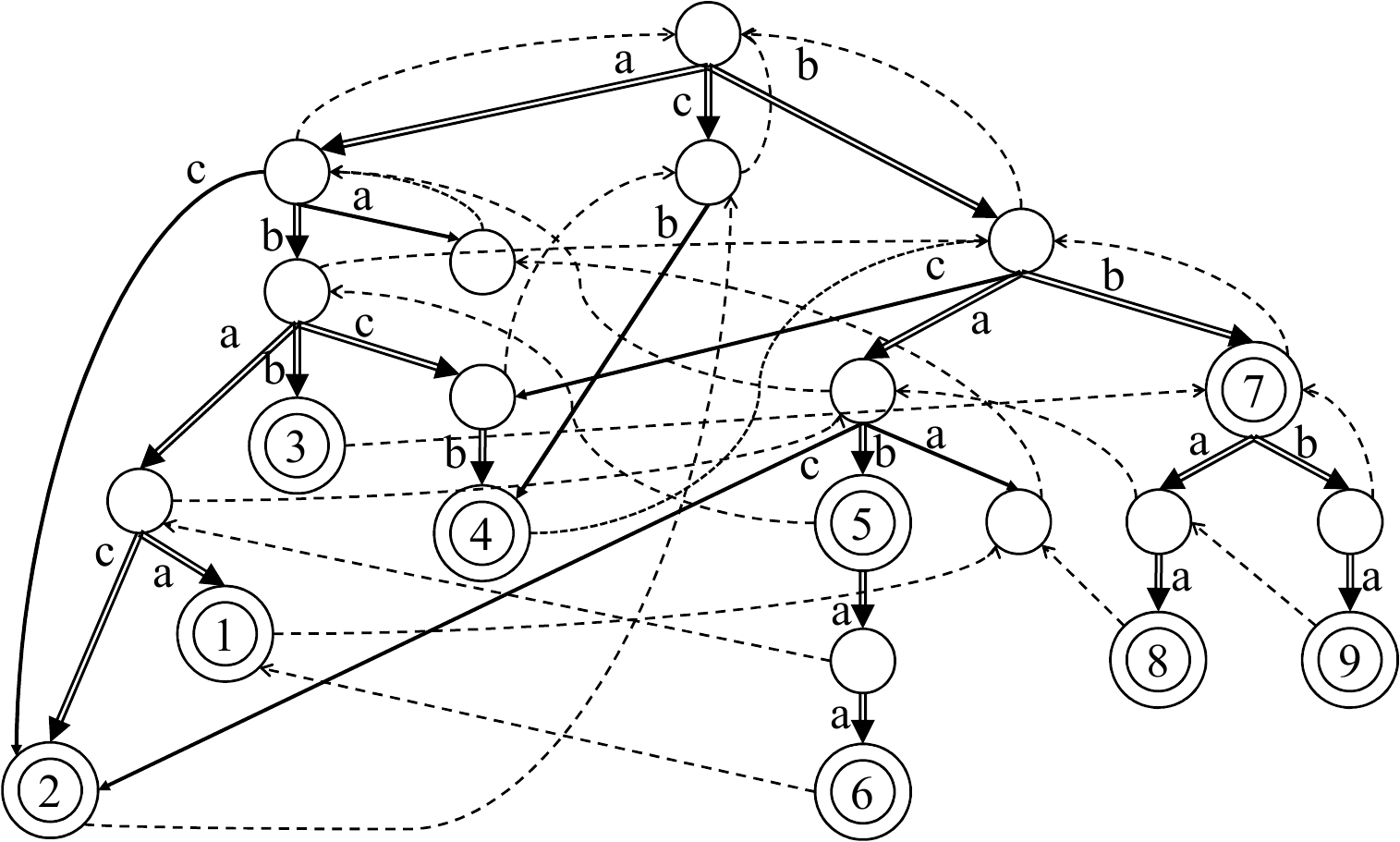}
  \caption{Illustration of the suffix links of $\DAWG(\mathcal{S})$ for the same set $\mathcal{S}$ of strings as in Fig~\ref{fig:DAWG}.}
  \label{fig:DAWG_suffixlinks}
\end{figure}

For a node $v \in V$ of $\DAWG(\mathcal{S})$,
let $\Long(v)$ denote the longest string represented by $v$
(i.e., $|\Long(v)|$ is the length of the longest path from the source to $v$).
Also, let $\slink(v) = u$ denote the suffix link from node $u$ to node $v$.

An edge $(u, a, v) \in E$ of $\DAWG(\mathcal{S})$
is called an \emph{primary edge}
if $|\Long(u)|+1 = |\Long(v)|$,
and it is called a \emph{secondary edge} otherwise ($|\Long(u)|+1 < |\Long(v)|$).
The primary edges form a spanning tree of $\DAWG(\mathcal{S})$
which consists of the longest paths from the source to all the nodes,
and let us denote this spanning tree by $\LPT(\mathcal{S})$.
By the definition of the equivalence class $\Eqc{\cdot}_{\mathcal{S}}$,
for each $S_i \in \mathcal{S}$ it holds that $ S_i= \Long(\Eqc{S_i}_{\mathcal{S}})$. Thus we have the following:
\begin{observation} \label{obs:LPT}
  For each $S_i \in \mathcal{S}$,
  the path that spells out $S_i$ from the source of $\DAWG(\mathcal{S})$
  is a path of $\LPT(\mathcal{S})$.
\end{observation}

$\DAWG(\mathcal{S})$ is a partial DFA of size $O(n)$ that accepts $\Substr(\mathcal{S})$~\cite{Blumer1985,Blumer1987}.
The following lemma permits us to efficiently update the DAWG
each time a new string is inserted to the current set of strings:
\begin{theorem}[\cite{Blumer1987}]\label{theo:DAWG_update}
  Suppose that $\DAWG(\mathcal{S}_{i-1})$ has been built
  for a set $\mathcal{S}_{i-1} = \{S_1, \ldots, S_{i-1}\}$ of $i-1$ strings.
  Given a new string $S_i$ to insert,
  one can update $\DAWG(\mathcal{S}_{i-1})$ to $\DAWG(\mathcal{S}_{i})$
  in $O(|S_i|\log \sigma)$ time.
\end{theorem}

\subsection{Suffix trees}

The \emph{suffix tree} for a set $\mathcal{S} = \{S_1, \ldots, S_k\}$ of strings,
denoted $\STree(\mathcal{S})$, is a compacted trie such that
\begin{itemize}
\item each edge is labeled by a non-empty substring in $\Substr(\mathcal{S})$;
\item the labels of the out-edges of each node begin with distinct characters;
\item every internal node is branching (i.e. has two or more children) or represents a suffix in $\Suffix(\mathcal{S})$;
\item every suffix in $\Suffix(\mathcal{S})$ is represented by a node.
\end{itemize}
For ease of explanations,
we identify each node of $\STree(\mathcal{S})$ with the string that the node represents.

The \emph{suffix link} of a node $v$ in $\STree(\mathcal{S})$ points to the node that represents the string $\str(v)[2..|\str(v)|]$.

See Fig.~\ref{fig:STree} and~\ref{fig:STree_suffixlinks} for
a concrete example of the suffix tree and its suffix links.

\begin{figure}[h!]
  \centering
  \includegraphics[scale=0.32]{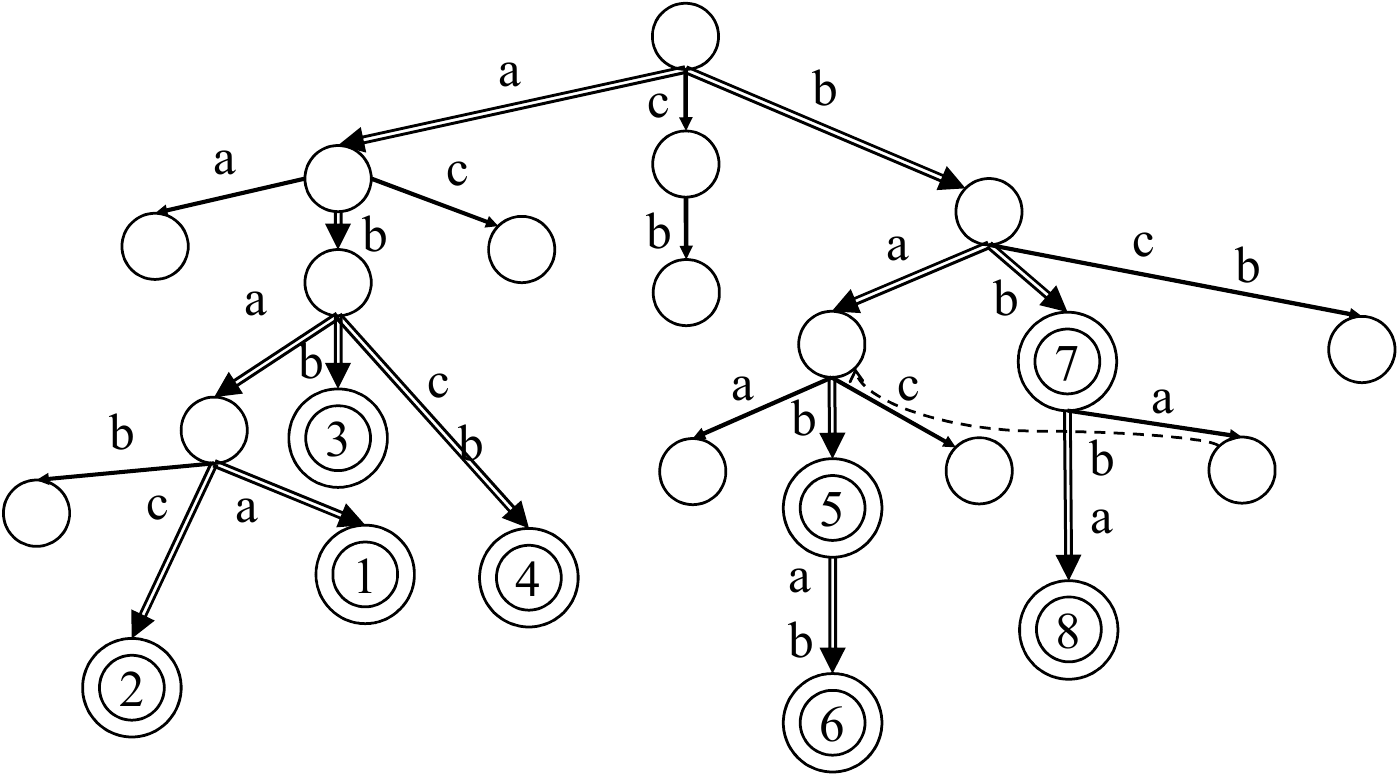}
  \caption{$\STree(\mathcal{S})$ for the same set $\mathcal{S} = \{\mathrm{abaa, abac, abb, abcb, bab, babaa, bb, bbaa, bbba}\}$ of strings as in Fig.~\ref{fig:AC}.
  The induced tree consisting only of the double-lined arcs is a compacted version of $\trie(\mathcal{S})$.}
  \label{fig:STree}
\end{figure}

\begin{figure}[h!]
  \centering
  \includegraphics[scale=0.32]{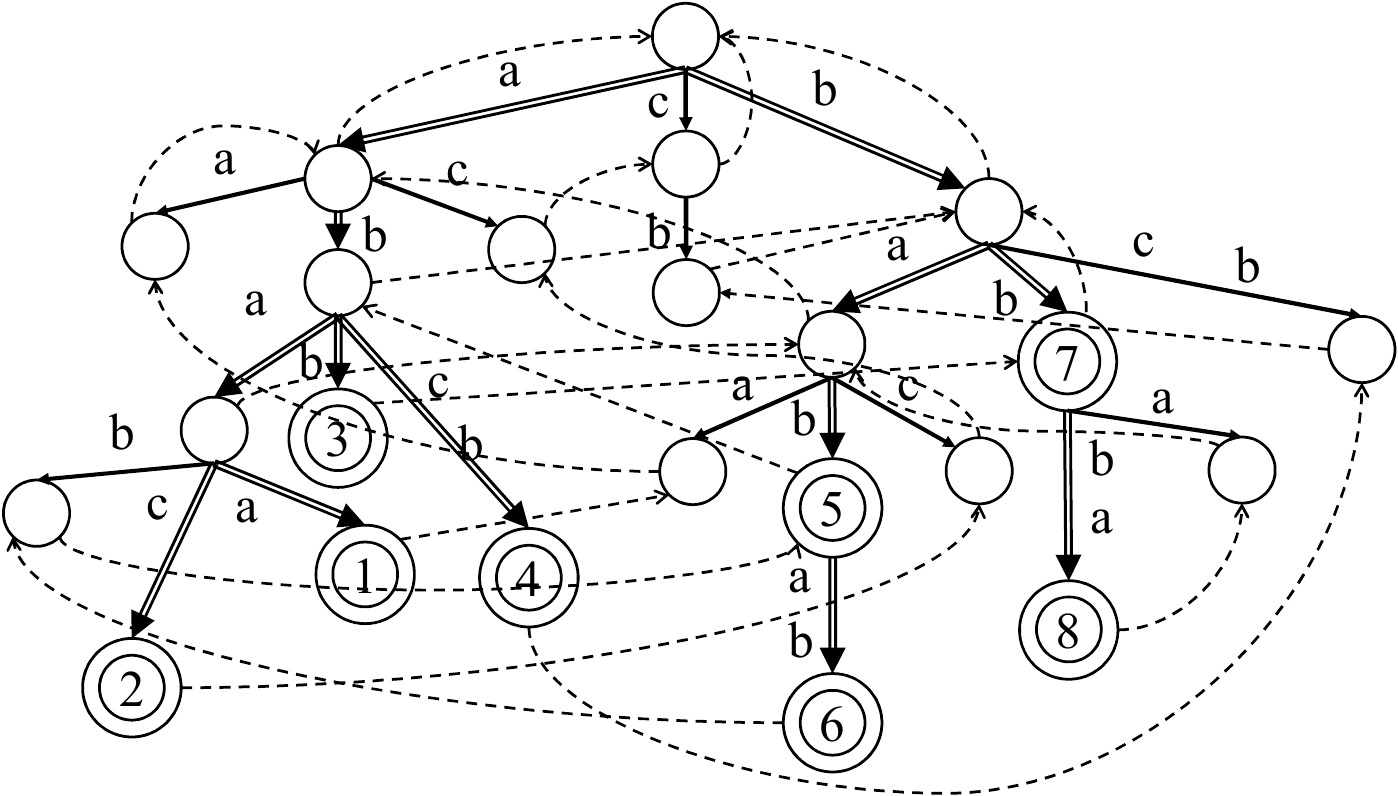}
  \caption{Illustration of the suffix links of $\STree(\mathcal{S})$ for the same set $\mathcal{S}$ of strings as in Fig~\ref{fig:STree}.}
  \label{fig:STree_suffixlinks}
\end{figure}

Each edge label $x$ of $\STree(\mathcal{S})$ is represented by
a tuple $(i,b,e)$ such that $x = S_i[b..e]$.
This way $\STree(\mathcal{S})$ can be stored in $O(n)$ space,
where $n = \Vert S \Vert$

Our version of suffix trees is the so-called \emph{Weiner tree}~\cite{Weiner73}
in which the repeated suffixes with unique right-extensions are represented by non-branching internal nodes
(For example, string $\mathrm{bab} \in \Suffix(\mathcal{S})$ in Fig.~\ref{fig:STree} is represented by a non-branching internal node).
This version of the suffix tree can easily be obtained by first building
the suffix tree for strings $S_1 \$_1, \ldots, S_k \$_k$ with
unique end-markers $\$_1, \ldots, \$_k$~($\$_i \neq \$_j$) and then
by removing all edges labeled by the end-markers $\$_1, \ldots, \$_k$.

By adapting Weiner's construction~\cite{Weiner73} or Ukkonen's construction~\cite{Ukkonen95} both working in online manners,
one can update the suffix tree efficiently when a new string is inserted to
the current set of strings.
This setting is called as the \emph{semi-dynamic} online construction
of suffix trees in the literature~\cite{TakagiIABH20}.

\begin{theorem}[\cite{Weiner73,Ukkonen95,TakagiIABH20}]\label{theo:STree_insert}
  Suppose that $\STree(\mathcal{S}_{i-1})$ has been built
  for a set $\mathcal{S}_{i-1} = \{S_1, \ldots, S_{i-1}\}$ of $i-1$ strings.
  Given a new string $S_i$ to insert,
  one can update $\STree(\mathcal{S}_{i-1})$ to $\STree(\mathcal{S}_{i})$
  in $O(|S_i|\log \sigma)$ time.
\end{theorem}

\section{Algorithm for Static-APSP} \label{sec:static}

In this section, we present a new, very simple algorithm for the static-APSP problem
for a static set $\mathcal{S} = \{S_1, \ldots, S_k\}$
of $k$ strings of total length $\Vert \mathcal{S} \Vert = n$.

Our main data structure is the AC-automaton $\AC(\mathcal{S})$.
%
For each string $S_i \in \mathcal{S}$ (in an arbitrary order),
set node $v$ to be the node that represents $S_i$,
i.e. initially $\str(v) = S_i$.
All the nodes of $\ctrie(\mathcal{S})$ are initially unmarked.
Our algorithm is shown in Algorithm 1.
See also Fig.~\ref{fig:traversal} for illustration.

\begin{itembox}[c]{Algorithm 1: Compute $\lspo(i,j)$ for all $1 \leq j \leq k$}
\begin{enumerate}
\item[(1)] Set $v$ to be the node that represents $S_i$ (i.e. $\str(v) = S_i$).
\item[(2)]
  \begin{itemize}
  \item Set $u \leftarrow \des(v)$.
  \item Traverse the \emph{induced} subtree $\mathsf{T}_u$ of $\ctrie(\mathcal{S})$ that is rooted at $u$, and consists only of unmarked nodes (and thus we do not traverse the subtrees under the marked nodes in $T_u$).
  \item Report $\lspo(i,j) = |\str(v)|$ for all id's $j$ found in the induced subtree $\mathsf{T}_u$.
  \item Mark node $u$ on $\ctrie(\mathcal{S})$.
  \end{itemize}
\item[(3)] If $v$ is the root, finish the procedure for the given string $S_i$. Unmark all marked nodes.
\item[(4)] If $v$ is not the root, take the failure link from $v$ on $\trie(\mathcal{S})$, and set $v \leftarrow \flink(v)$. Continue Step (2).
\end{enumerate}
\end{itembox}

\begin{figure}[h]
  \centering
  \includegraphics[scale=0.4]{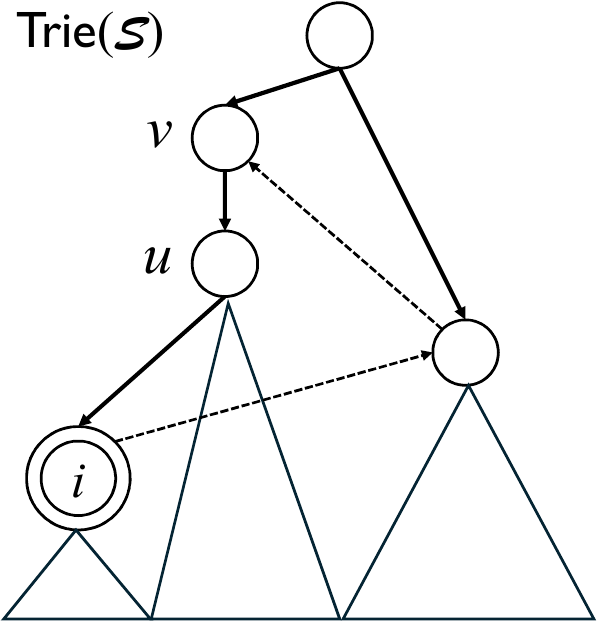}
  \hfil
  \includegraphics[scale=0.4]{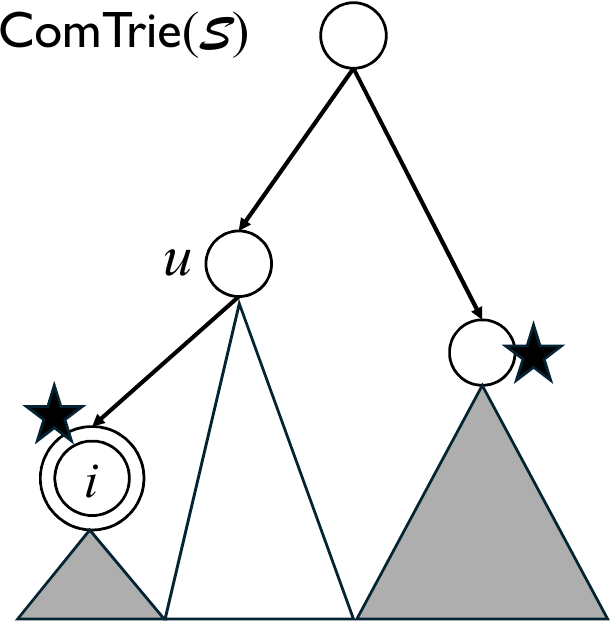}
  \caption{Illustration for Algorithm 1 that solves the static-APSP problem, in which we have climbed up the failure links from the node with id $i$ to node $v$. We access $u = \des(v)$ in $\ctrie(\mathcal{S})$, and traverse only the white subtree below $u$.}
  \label{fig:traversal}
\end{figure}

\begin{theorem} \label{theo:static}
  For a static set $\mathcal{S} = \{S_1, \ldots, S_k\}$ of $k$ strings
  of total length $n$,
  our algorithm solves the static-APSP problem with $O(n)$ working space and 
  \begin{itemize}
  \item in $O(n + k^2)$ time for an integer alphabet of size $\sigma = n^{O(1)}$;
  \item in $O(n \log \sigma + k^2)$ time for a general alphabet of size $\sigma$.
  \end{itemize}
\end{theorem}

\begin{proof}
  The correctness of our algorithm follows from the definitions of the $\lspo(\cdot,\cdot)$ function and the failure links of $\AC(\mathcal{S})$,
  and from our failure link traversal starting from the node $v$ representing $S_i$, in decreasing order of the node depths towards the root $r$.
  
  For each string $S_i$, Step (1) clearly takes $O(1)$ time.
  The total number $l$ of nodes $v, \flink(v), \flink^2(v), \ldots, \flink^{l-1}(v) = r$ visited in the chain of failure links starting from node $v$ with $\str(v) = S_i$ is at most $|S_i|+1$.
  The size of $\ctrie(\mathcal{S})$ is $O(k)$.
  As we process the nodes $v, \flink(v), \flink^2(v), \ldots, \flink^{l-1}(v) = r$ in decreasing order of their depths,
  any marked node is visited at most once.
  This is because, if a marked node $w$ is visited twice,
  then there must be two distinct ancestors $p,q$ of $w$ in $\ctrie(\mathcal{S})$, such that $p$ is an ancestor of $q$.
  This however is a contradiction,
  since $q$ was marked before $p$,
  and our algorithm does not visit the subtree under $q$
  when processing the induced tree $\mathsf{T}_{p}$.
  Therefore, each node is marked at most once,
  and each marked node is visited at most once.
  Thus, the traversals on $\ctrie(\mathcal{S})$ with our marking operations
  takes time linear in its size $O(k)$.
  Overall, for each $S_i \in \mathcal{S}$,
  we can compute $\lspo(i,j)$ for all $1 \leq i \leq k$
  in optimal $O(|S_i| + i)$ time,
  which leads to an overall $O(n + k^2)$ time for all the $k$ strings in $\mathcal{S}$.

  We build $\AC(\mathcal{S})$ by Theorem~\ref{theo:AC_const}.
  We can compute $\ctrie(\mathcal{S})$ and $\des(u) = v$ for all nodes $v$ in $\trie(\mathcal{S})$ in time linear in the size of $\trie(\mathcal{S})$,
  which is $O(n)$.
\end{proof}

For the length-bounded APSP problem, we have the following:
\begin{corollary} \label{coro:static_length}
For a static set $\mathcal{S} = \{S_1, \ldots, S_k\}$ of $k$ strings of total length $n$,
our algorithm solves the length-$\ell$ static-APSP problem with $O(n)$ working space and 
  \begin{itemize}
  \item in $O(n + |\mathcal{L_{\mathcal{S}}^{\geq \ell}}|)$ time for an integer alphabet of size $\sigma = n^{O(1)}$;
  \item in $O(n \log \sigma + |\mathcal{L_{\mathcal{S}}^{\geq \ell}}|)$ time for a general alphabet of size $\sigma$.
  \end{itemize}
Our data structure can be reused for different values of the length-threshold $\ell$.
\end{corollary}

\begin{proof}
  It suffices for us to change Step (3) and Step (4) in Algorithm 1 as follows:
  \begin{enumerate}
  \item[(3)] \underline{If $|\str(v)| \leq \ell$}, finish the procedure for the given string $S_i$. Unmark all marked nodes.
  \item[(4)] \underline{If $|\str(v)| > \ell$}, take the failure link from $v$ on $\trie(\mathcal{S})$, and set $v \leftarrow \flink(v)$. Continue Step (2).
  \end{enumerate}
  We can reuse our data structure for different values of $\ell$,
  as we only need to change the underlined conditions in the above.
\end{proof}

\section{Algorithm for Dynamic APSP} \label{sec:dynamic}

In this section, we present our algorithm for solving the dynamic APSP problem:
Given a new string $S_i$ that is inserted to the current set $\mathcal{S}_{i-1} = \{S_1, \ldots, S_{i-1}\}$ of $i-1$ strings,
our goal is to compute the outputs
$\mathcal{F}_i$ and $\mathcal{B}_i$ of Problem~\ref{prob:dynamic}.

In what follows, we will show the following:
\begin{theorem} \label{theo:dynamic}
  There exists a DAWG-based data structure of $O(n_i)$ space,
  which
  solves the dynamic APSP problem of 
  computing $\mathcal{F}_{i}$ and $\mathcal{B}_i$
  for a given new string $S_{i}$ in $O(|S_i| \log \sigma + i)$ time,
  where $n_i = \Vert \mathcal{S}_{i}\Vert = \sum_{j=1}^i|S_i|$.
\end{theorem}

\begin{corollary} \label{coro:dynamic_length}
  Let $\ell > 0$ be a length threshold.
  There exists a DAWG-based data structure of $O(n_i)$ space,
  which solves the length-bounded dynamic APSP problem of 
  computing $\mathcal{F}_{i}^{\geq \ell}$ and $\mathcal{B}_i^{\geq \ell}$
  for a given new string $S_{i}$ in $O(|S_i| \log \sigma + i)$ time,
  where $n_i = \Vert \mathcal{S}_{i}\Vert = \sum_{j=1}^i|S_i|$.
\end{corollary}

We will focus on proving Theorem~\ref{theo:dynamic}.
Corollary~\ref{coro:dynamic_length} can be obtained from Theorem~\ref{theo:dynamic},
in a similar manner to the static case in Section~\ref{sec:static}.

\subsection{Computing $\mathcal{F}_i$}
\label{sec:forward}


We use the DAWG data structure for the dynamic case.
Consider a set $\mathcal{S} = \{S_1,\ldots, S_{i-1}\}$ of $i-1 $ strings.
By Observation~\ref{obs:LPT},
$\trie(\mathcal{S}_{i-1})$ is an induced tree of $\LPT(\mathcal{S}_{i-1})$
that consists only of the paths spelling out $S_1, \ldots, S_{i-1}$
from the source of $\DAWG(\mathcal{S}_{i-1})$.
Also, if $\flink(u) = v$ for a node $u$ of $\trie(\mathcal{S}_{i-1})$,
then there is a chain of suffix links $u, \slink(u), \ldots, v$
of length $k \geq 1$
from the node $u$ to the node $v$ in $\DAWG(\mathcal{S}_{i-1})$
(i.e. $\slink^k(u) = v$).
It is known that the suffix links of $\DAWG(\mathcal{S}_{i-1})$
form an edge-reversed tree that is equivalent to the suffix tree~\cite{Weiner73} of
the \emph{reversed} strings~\cite{Blumer1987,TakagiIABH20}.
We denote this suffix link tree by $\SLT(\mathcal{S})$,
and will use $\SLT(\mathcal{S})$ in the dynamic case,
instead of the tree $\FLT(\mathcal{S})$ of failure links in the static case.

Now, we design our algorithm for the dynamic case with the DAWG structure,
using the idea from Section~\ref{sec:static} for the static case with the AC-automaton (see also Algorithm 1).
Suppose that we have built and maintained the aforementioned data structures
for the dynamic set $\mathcal{S}_{i-1} = \{S_1, \ldots, S_{i-1}\}$ of $i-1$ strings.
Given a new string $S_i$ to insert, we update $\DAWG(\mathcal{S}_{i-1})$
to $\DAWG(\mathcal{S}_{i})$ in $O(|S_i|\log \sigma)$ time with Theorem~\ref{theo:DAWG_update}.
We have now updated $\SLT(\mathcal{S}_{i-1})$ to $\SLT(\mathcal{S}_i)$ as well.

To perform equivalent procedures as in Algorithm 1,
we climb up the suffix link path from the node $v$ with $\str(v) = S_i$ toward the root on $\SLT(\mathcal{S}_i)$.
Namely, we use $\slink$ instead of $\flink$.
Each time we visit a node $v$ that is contained in $\trie(\mathcal{S}_i)$,
we access $u = \des(v)$ on $\ctrie(\mathcal{S}_i)$
and perform the same procedures as in Algorithm 1.
This gives us $\lspo(i,j)$ for all $1 \leq j \leq i$
in a total of $O(|S_i| \log \sigma + i)$ time, since $\ctrie(\mathcal{S}_i)$ is of size $O(i)$.
We can determine whether a node $v$ is in $\trie(\mathcal{S}_i)$ or not in $O(1)$ time
after a simple $O(|S_i| \log \sigma)$-time preprocessing per new string $S_i$:
After inserting $S_i$ to the DAWG,
we simply trace the path that spells out $S_i$ from the source on the DAWG.

What remains is how to update $\ctrie(\mathcal{S}_{i-1})$ to $\ctrie(\mathcal{S}_i)$,
and update the $\des$ function.
We do this by na\"ively inserting the new string $S_i$ to $\ctrie(\mathcal{S}_{i-1})$,
traversing the tree with $S_{i}$ from the root in $O(|S_i| \log \sigma)$ time.
The $\des$ function needs to be updated when
an edge $(v, u)$ on $\ctrie(\mathcal{S}_{i-1})$ is split into two edges $(v,w)$ and $(w,u)$,
where $w$ is the parent of the new leaf representing $S_i$ on $\ctrie(\mathcal{S}_i)$.
Let $v = v_1, \ldots, v_x = w$ be the nodes in the uncompacted trie
$\trie(\mathcal{S}_i)$
that are implicit on the edge $(v,w)$ in the compacted trie $\ctrie(\mathcal{S}_i)$.
We update $\des(v_1) \leftarrow w, \ldots, \des(v_{x}) \leftarrow w$.
We set $\des(y) \leftarrow s$ for all new nodes $y$ of $\trie(\mathcal{S}_i)$,
where $s$ is the new leaf representing $S_i$.
These updates can be done in $O(|S_i|)$ time.

To summarize this subsection, we have proven the following:
\begin{lemma} \label{lem:forward}
  There is a DAWG-based data structure of $O(n_i)$ space,
  which computes $\mathcal{F}_i$ in $O(|S_i| \log \sigma + i)$ time per new string $S_i$
  inserted.
\end{lemma}

\begin{remark}
  Diptarama et al.~\cite{HendrianIYS19} showed how to simulate
  the failure link tree of the AC-automaton
  with the suffix link tree of the DAWG that is augmented with
  a nearest marked ancestor (NMA) data structure~\cite{Westbrook92}.
  It should be noted that NMA data structures are not required in our algorithm.
\end{remark}

\subsection{Computing $\mathcal{B}_i$}
\label{sec:reverse}

We can easily reduce the problem of computing $\mathcal{B}_i$
to computing $\mathcal{F}_i$ for the set
$\rev{\mathcal{S}_i} = \{\rev{S_1}, \ldots, \rev{S_i}\}$ of \emph{reversed strings},
where the roles of suffixes and prefixes are swapped.
This immediately gives us the following:
\begin{corollary}
  There is a DAWG-based data structure of $O(n_i)$ space,
  which computes $\mathcal{B}_i$ in $O(|S_i| \log \sigma + i)$ time
  per new string $S_i$ inserted.
\end{corollary}

\section{Algorithm for Fully Dynamic APSP}
\label{sec:fully_dynamic}

In this section,
we present a suffix tree based algorithm
for the fully dynamic APSP problem allowing for insertions and deletions of strings.

We make use of the following observation in our algorithm.

\begin{observation}
  \label{sec:suffix_tree_induced}
  The induced tree of $\STree(\mathcal{S})$
  that consists only of the paths spelling out the strings in $\mathcal{S}$
  is a compacted version of $\trie(\mathcal{S})$.
  The set of nodes of $\ctrie(\mathcal{S})$ is
  a subset of the nodes of this induced tree.
\end{observation}
See Fig.~\ref{fig:AC} and Fig.~\ref{fig:STree} for a concrete example of this observation.

\subsection{Insertions}

In the case of insertions, one can use essentially the same algorithm from Section~\ref{sec:dynamic}:
When a new string $S_i$ is inserted to the current set $\mathcal{S}$,
we first update the suffix tree in $O(|S_i|\log \sigma)$ time by Theorem~\ref{theo:STree_insert}.
Then, we climb up the suffix link path from the node representing $S_i$ toward the root.
By Observation~\ref{sec:suffix_tree_induced},
the $\des(\cdot)$ function to access the corresponding nodes in $\ctrie(\mathcal{S})$ can be maintained analogously.
In more detail, by Observation~\ref{sec:suffix_tree_induced},
a single edge in $\ctrie(\mathcal{S})$ may be split into
multiple edges in the compacted version of $\trie(\mathcal{S})$.
Thus, 
when an edge $(u,v)$ in the compacted version of $\trie(\mathcal{S})$
is split into two edges $(u,w)$ and $(w,v)$ during the update
of the suffix tree,
then we simply set $\des(w) \leftarrow \des(v)$.
These all give us how to compute $\mathcal{F}_i$ in $O(|S_i| \log \sigma)$ time
using $\STree(\mathcal{S})$ of $O(n)$ space,
where $n = \Vert \mathcal{S} \Vert$,
as in Section~\ref{sec:forward}.

One can compute $\mathcal{B}_i$ in $O(|S_i| \log \sigma)$ time
and $O(n)$ space
by using $\STree(\rev{\mathcal{S}})$ for the set $\rev{\mathcal{S}}$
of reversed strings, as in Section~\ref{sec:reverse}.

\subsection{Deletions}

We can also extend our suffix-tree based algorithm
in the case where both insertions and deletions are supported.

Suppose that a string $S_i$ is deleted from the current set $\mathcal{S}$ of strings.
Let $S_i[r..|S_i|]$ be the longest suffix of $S_i$
that corresponds to a branching node in the suffix tree,
and let $S_i[t..|S_i|]$ be the longest suffix of $S_i$
such that $S_i[t..|S_i|] = S_j$ for some $j \in \ind(\mathcal{S}) \setminus \{i\}$.
Let $t = \infty$ when there is no such $j$ exists.
Let $p = \min\{r,t\}$.

For each $1 \leq q < p$ in increasing order,
we remove the nodes $v_q$ (leaves or non-branching internal nodes)
that respectively represent the suffixes $S_i[q..|S_i|]$.
By tracing the suffix link chain from $v_1$ to $v_{p}$,
this can be done in $O(1)$ time for each suffix $S_i[q..|S_i|]$
that is removed from the tree.
If the parent of $v_q$ becomes non-branching
after the removal of $v_q$, then the parent node is removed and its in-edge and out-edge are merged into a single edge.
This is the same procedure as in the \emph{suffix tree in the sliding window model}~\cite{Larsson1996,Senft2005,LeonardIMB2024},
and the total time complexity is $O(p \log \sigma) \subseteq O(|S_i| \log \sigma)$ for this step.
Now the tree topology of the suffix tree has been updated.

Recall that each edge label $x$ is represented by a tuple
$(i,b,e)$ such that $x = S_i[b..e]$.
When $S_i$ is removed from $\mathcal{S}$,
we need to replace a tuple $(i,b,e)$ which refers to $S_i$
with some tuple $(i',b',e')$ such that $i' \in \ind(\mathcal{S}) \setminus \{i\}$ and $x = S_{i'}[b'..e']$.
This can be done in $O(1)$ time by adapting the leaf-pointer algorithm~\cite{LeonardIMB2024} originally proposed for the sliding suffix trees.
Leonard et al.~\cite{LeonardIMB2024} showed that
edge-label maintenance in the suffix tree reduces to
maintaining a pointer from each node to an arbitrary leaf in its subtree.
They showed how to maintain such pointers in worst-case $O(1)$ time per leaf deletion.
By applying this technique,
we can update edge labels of $\STree(\mathcal{S})$ in a total of $O(p) \subseteq O(|S_i|)$ time.

Analogously, we update $\STree(\rev{\mathcal{S}})$
in $O(|S_i|\log \sigma)$ time for the deleted string $S_i$.

In this fully dynamic setting,
we maintain the output list $\mathcal{L}_{\mathcal{S}}$ with
two arrays $F$ and $B$
such that $F[i]$ stores $\lspo(i,j)$ for all $j \in \ind(\mathcal{S})$
and $B[i]$ stores $\lspo(j,i)$ for all $j \in \ind(\mathcal{S})$.
When a string $S_i$ is deleted,
then we simply remove all elements in $F[i]$ and in $B[i]$.
This clearly takes $O(k)$ time, where $k = |\mathcal{S}|$.
We maintain these two dynamic arrays $F$ and $B$ by a standard
doubling/halving technique:
When all the entries of $F$ and $B$ become full,
then we rebuild two arrays with $2|F| = 2|B|$ entries each.
When the number of occupied entries in $F$ and in $B$ reaches $|F|/4 = |B|/4$,
then we rebuild two arrays with $|F|/2 = |B|/2$ entries each,
and refresh the string ids.
This takes $O(k)$ amortized time per insertion/deletion of a string,
where $k$ is the number of strings in the set $\mathcal{S}$.

Finally, we obtain the following results:

\begin{theorem} \label{theo:fully_dynamic}
  There exists a suffix-tree-based data structure of $O(n)$ space,
  which solves the fully dynamic APSP problem  
  in $O(|S_i| \log \sigma + k)$ amortized time
  for each inserted/deleted string $S_i$,
  where $n = \Vert \mathcal{S} \Vert$ and $k = |\mathcal{S}|$.
\end{theorem}

\begin{corollary} \label{coro:fully_dynamic_length}
  Let $\ell > 0$ be a length threshold.
  There exists a suffix-tree-based data structure of $O(n)$ space,
  which solves the fully dynamic APSP problem 
  in $O(|S_i| \log \sigma + k)$ amortized time
  for each inserted/deleted string $S_i$,
  where $n = \Vert \mathcal{S} \Vert$ and $k = |\mathcal{S}|$.
\end{corollary}

\section{Conclusions and Future Work}

In this paper, we presented an efficient algorithm
for the APSP problem for a dynamic set of strings
$\mathcal{S} = \{S_1, \ldots, S_k\}$ of total length $n$.
Our algorithm occupies $O(n)$ space and takes 
$O(|S_i| \log \sigma + k)$ time for each string $S_i$ inserted or deleted.
The running time is near-optimal except for the $\log \sigma$ factor,
that is indeed required if one maintains $\mathcal{S}$ in a tree or DAG based data structure for a general ordered alphabet of size $\sigma$.


It is interesting to see whether our (fully) dynamic algorithm can be
extended to efficient construction/update of the HOG~\cite{CanovasCR17,Khan21,ParkPCPR21} for a (fully) dynamic set of strings.

\section*{Acknowledgements}
This work was supported by JSPS KAKENHI Grant Number JP23K24808~(SI).

\bibliographystyle{abbrv}
\bibliography{ref}

\end{document}